\documentclass[a4paper,11pt]{article}

\usepackage[top=1in,bottom=1in,left=1in,right=1in]{geometry}
\usepackage{amsmath,amssymb,amsthm,mathrsfs}
\usepackage{booktabs}
\usepackage{microtype}
\usepackage{xcolor}
\usepackage{cite}
\usepackage{tikz}
\usetikzlibrary{arrows.meta,positioning}
\usepackage{hyperref}
\hypersetup{colorlinks=true,linkcolor=blue,citecolor=blue,urlcolor=blue}
\numberwithin{equation}{section}
\theoremstyle{plain}
\newtheorem{theorem}{Theorem}[section]
\newtheorem{proposition}[theorem]{Proposition}
\newtheorem{lemma}[theorem]{Lemma}
\newtheorem{corollary}[theorem]{Corollary}

\theoremstyle{definition}
\newtheorem{definition}[theorem]{Definition}

\theoremstyle{remark}
\newtheorem{remark}[theorem]{Remark}

\newcommand{\Id}{I}
\newcommand{\Mat}{\operatorname{Mat}}
\newcommand{\GL}{\operatorname{GL}}
\newcommand{\SL}{\operatorname{SL}}
\newcommand{\tr}{\operatorname{tr}}
\newcommand{\diag}{\operatorname{diag}}

\newcommand{\wt}{\operatorname{wt}}

\newcommand{\bbC}{\mathbb C}
\newcommand{\bbZ}{\mathbb Z}
\newcommand{\bt}{\boldsymbol t}
\newcommand{\e}{\mathrm e}

\title{%
  \fontsize{15}{20}\selectfont\bfseries Sato-theoretic construction of the anti-self-dual Yang--Mills 
  \\ hierarchy and the Ward conjecture
}
\author{Shangshuai Li$^{1,2}$\thanks{Email: lishangshuai@nbu.edu.cn}
	~~Da-jun Zhang$^{3,4}$\thanks{Corresponding author. Email: djzhang@staff.shu.edu.cn}
  \\[5pt]
  {\small $^1$ School of Mathematics and Statistics, Ningbo University, Ningbo 315211, China} \\
  {\small $^2$ Graduate School of Mathematics, Nagoya University, Nagoya 464-8602, Japan} \\
  {\small $^3$ Department of Mathematics, Shanghai University, Shanghai 200444, China} \\
  {\small $^4$ Newtouch Center for Mathematics of Shanghai University, Shanghai 200444, China}
}
\date{\today}

\begin{document}
\maketitle

\begin{abstract}
	
We develop the Sato-theoretic dressing framework for the anti-self-dual Yang--Mills (ASDYM) hierarchy based on a normalized Riemann--Hilbert decomposition. A four-sector expansion of the generating function yields a bi-infinite matrix array of relative coordinates, extending the affine coordinates on the big cell of the Sato Grassmannian. We derive the Sato--Wilson equations, continuous coordinate flows, and their discrete analogues. Several classical integrable hierarchies are recovered under dimensional reduction constraints, with their nonlinear variables identified as specific relative coordinates.

\end{abstract}



\begingroup
\small
\tableofcontents
\endgroup

\newpage
\section{Introduction}
\label{sec:introduction}

The anti-self-dual Yang--Mills (ASDYM) equations occupy a central position in mathematical physics \cite{Yang1977}. In twistor theory, the Penrose--Ward correspondence identifies anti-self-dual gauge fields with holomorphic vector bundles over twistor space. Restriction to an individual twistor line reduces the reconstruction of the gauge field to a matrix splitting problem for the associated patching matrix \cite{AtiyahWard1977,MasonWoodhouse1996,Ward1977}. From the viewpoint of integrable systems, this splitting problem is a Riemann--Hilbert problem, or equivalently a Birkhoff factorization. Its local factors are holomorphic near the distinguished points \(\lambda=0\) and \(\lambda=\infty\) and can be interpreted as dressing matrices \cite{Takasaki1984,Takasaki1985,UenoNakamura1982,UenoNakamura1983}. This places the patching function within the loop-group factorization framework underlying Sato theory \cite{PressleySegal1986,SegalWilson1985} and provides a natural setting for the ASDYM hierarchy \cite{KakeiIkedaTakasaki2002,Nakamura1988,Takasaki1990}.

Much of the significance of the ASDYM equations stems from their rich family of reductions to lower-dimensional integrable systems. In 1985, Ward proposed that perhaps all integrable systems could be obtained as reductions of the ASDYM equations \cite{Ward1985}. This proposal, now known as the Ward conjecture, has motivated a broad program of explicit reductions. Such reductions have been constructed through Lax representations \cite{AblowitzChakravartyHalburd2003,AblowitzChakravartyTakhtajan1993}, twistor methods \cite{MasonSparling1989,MasonWoodhouse1996,Strachan1993}, bilinear methods \cite{SasaOhtaMatsukidaira1998}, and direct reductions of nonlinear matrix equations \cite{GeWangWu1994,LiLiuZhang2025,LiMarunoZhang2026,BenincasaHalburd2016}.

Despite the many existing approaches, reductions developed in different formulations may take rather different forms and appear unrelated at first sight. Some methods are mutually compatible, whereas others are formulated under different conventions, making a systematic study of the Ward conjecture difficult. Based on the Sato-theoretic construction of the ASDYM hierarchy \cite{Takasaki1984,KakeiIkedaTakasaki2002} and the direct reduction scheme based on ASDYM matrix variables \cite{LiMarunoZhang2026,LiLiuZhang2025}, we aim to develop a unified framework for studying the Ward conjecture in this paper.

The main results of this paper can be summarized as follows.
\begin{enumerate}
	\item Within the Sato--theoretic construction for the ASDYM hierarchy, we extend the affine coordinates on the big cell of the Sato Grassmannian to a bi-infinite matrix array, whose entries we call relative coordinates.
	
	\item We revisit the continuous Sato--Wilson equations, auxiliary linear system, and coordinate flows. By employing Miwa shifts, we derive their discrete analogues, thereby unifying the continuous and discrete reductions of the ASDYM hierarchy within a single framework.
	
	\item The Sato-theoretic origin of the dimensional reduction constraint employed in \cite{LiMarunoZhang2026,LiLiuZhang2025} is identified. This reduction recovers several classical integrable systems, with their nonlinear variables expressible in terms of specific relative coordinates.
\end{enumerate}

The remainder of this paper is organized as follows. Starting from a normalized Riemann--Hilbert decomposition, Section~\ref{sec:dressing-matrices-relative-coordinates} introduces the dressing matrices and relative coordinates. Section~\ref{sec:sato-theoretic-construction} incorporates the time variables and develops the corresponding Sato-theoretic framework. In Section~\ref{sec:Ward-conjecture}, this framework is applied to derive several representative continuous and discrete integrable systems. Finally, Section~\ref{sec:concluding-remarks} concludes the paper with a brief discussion.

\section{Dressing matrices and relative coordinates}
\label{sec:dressing-matrices-relative-coordinates}

The Riemann--Hilbert problem provides the analytic origin of the dressing matrices. Their local expansions at \(\lambda=0\) and \(\lambda=\infty\) are then combined into a single bi-infinite array by a two-parameter generating function.

\subsection{Riemann--Hilbert problem and dressing matrices}

Fix constants \(0<r<R\), and set
\[
D_0:=\{\lambda\in\bbC:|\lambda|<R\},\qquad
D_\infty:=\{\lambda\in\bbC\cup\{\infty\}:|\lambda|>r\}.
\]
Let \(g(\lambda)\) be a holomorphic \(\GL_N(\bbC)\)-valued transition function on \(D_0\cap D_\infty\), and suppose that it admits the normalized Riemann--Hilbert decomposition
\begin{equation}
	g(\lambda)=W_0(\lambda)^{-1}W_\infty(\lambda),
	\label{eq:birkhoff-factorization}
\end{equation}
where \(W_0\) and \(W_\infty\) are holomorphic and invertible on \(D_0\) and \(D_\infty\), respectively, with the normalization condition \(W_\infty(\infty)=\Id_N\). Their local expansions have the form
\begin{equation}
	W_\infty(\lambda)=\Id_N-\sum_{n\geq1}w_n\lambda^{-n},\qquad
	W_0(\lambda)=\widehat w_0+\sum_{n\geq1}\widehat w_n\lambda^n.
	\label{eq:local-dressing-expansions}
\end{equation}
Set \(\bbZ_0:=\bbZ_{<0}\), \(\bbZ_\infty:=\bbZ_{\geq0}\), \(\varepsilon_0:=-1\), and \(\varepsilon_\infty:=1\). The coefficients of \eqref{eq:local-dressing-expansions} determine a sequence \(\{w_{0j}\}_{j\in\bbZ}\subset\Mat_N(\bbC)\) through
\begin{equation}
	W_a(\lambda)=\Id_N-\varepsilon_a\sum_{j\in\bbZ_a}w_{0j}\lambda^{-j-1},\qquad a\in\{0,\infty\}.
	\label{eq:dressing-row-expansion}
\end{equation}
Explicitly,
\[
w_{0,-1}=\widehat w_0-\Id_N,\qquad
w_{0,-n-1}=\widehat w_n,\qquad
w_{0,n-1}=w_n,\qquad n\geq1.
\]
For consistency with the bi-infinite array introduced later,
we call this sequence the zeroth row. Similarly, the local expansions of the inverses define the zeroth column \(\{w_{i0}\}_{i\in\bbZ}\subset\Mat_N(\bbC)\) by
\begin{equation}
	W_a(\lambda)^{-1}=\Id_N+\varepsilon_a\sum_{i\in\bbZ_a}w_{i0}\lambda^{-i-1},\qquad a\in\{0,\infty\}.
	\label{eq:dressing-column-expansion}
\end{equation}

\subsection{Relative coordinates and recurrence relation}

We now combine the zeroth row and column into a bi-infinite matrix array.

\begin{proposition}[Generating function]
	\label{prop:kernel-expansion}
	Let \(\eta\) and \(\zeta\) be independent spectral parameters, and set \(u:=\eta^{-1}\) and \(v:=\zeta^{-1}\). Define
	\[
	\mathscr M_{ab}:=
	\begin{cases}
		\Mat_N\bigl(\bbC\{\eta,\zeta\}\bigr),&(a,b)=(0,0),\\[1mm]
		v\Mat_N\bigl(\bbC\{\eta,v\}\bigr),&(a,b)=(0,\infty),\\[1mm]
		u\Mat_N\bigl(\bbC\{u,\zeta\}\bigr),&(a,b)=(\infty,0),\\[1mm]
		uv\Mat_N\bigl(\bbC\{u,v\}\bigr),&(a,b)=(\infty,\infty),
	\end{cases}
	\]
	where \(\bbC\{x,y\}\) denotes the ring of convergent power series at \((x,y)=(0,0)\). For \(a,b\in\{0,\infty\}\), the pair \((a,b)\) labels the sector in which \(\eta\) is expanded at \(a\) and \(\zeta\) at \(b\). Define 
	\begin{equation}
		\mathcal W_{ab}(\eta,\zeta):=
		\frac{W_a(\eta)^{-1}W_b(\zeta)-\Id_N}{\zeta-\eta}.
		\label{eq:relative-kernel-identity}
	\end{equation}
	Then \(\mathcal W_{ab}(\eta,\zeta)\in\mathscr M_{ab}\), and  the following boundary identities hold
	\begin{equation}
		[\zeta^{-1}]\mathcal W_{a\infty}(\eta,\zeta)=W_a(\eta)^{-1}-\Id_N,
		\qquad
		[\eta^{-1}]\mathcal W_{\infty b}(\eta,\zeta)=\Id_N-W_b(\zeta).
		\label{eq:boundary-condition}
	\end{equation}
\end{proposition}

\begin{proof}
	For later use, denote the numerator by
	\[
	N_{ab}(\eta,\zeta):=
	W_a(\eta)^{-1}W_b(\zeta)-\Id_N.
	\]
	We treat the four sectors separately. Whenever a spectral parameter is
	expanded at infinity, we use the corresponding local coordinates
	$u:=\eta^{-1}$ and $v:=\zeta^{-1}$.
	
	In the \((0,0)\)-sector, \(N_{00}\) is holomorphic near \((0,0)\) and vanishes for \(\zeta=\eta\). Hence
	\[
	N_{00}(\eta,\zeta)=(\zeta-\eta)H_{00}(\eta,\zeta),
	\qquad H_{00}\in\Mat_N\bigl(\bbC\{\eta,\zeta\}\bigr),
	\]
	thus \(\mathcal W_{00}=H_{00}\in\mathscr M_{00}\). Likewise, in the $(\infty,\infty)$-sector, one has
	\[
	N_{\infty\infty}(u,v)=(u-v)H_{\infty\infty}(u,v),
	\qquad H_{\infty\infty}\in\Mat_N\bigl(\bbC\{u,v\}\bigr).
	\]
	Since \(\zeta-\eta=(u-v)/(uv)\), it follows that \(\mathcal W_{\infty\infty}=uvH_{\infty\infty}\in\mathscr M_{\infty\infty}\).
	
	In the mixed sectors, 
	\[
	\frac{1}{\zeta-\eta}=\frac{v}{1-\eta v}=-\frac{u}{1-u\zeta}.
	\]
	The factors \((1-\eta v)^{-1}\) and \((1-u\zeta)^{-1}\) are convergent near the origin. Hence
	\[
	\mathcal W_{0\infty}=\frac{v}{1-\eta v}N_{0\infty}(\eta,v)\in\mathscr M_{0\infty},
	\qquad
	\mathcal W_{\infty0}=-\frac{u}{1-u\zeta}N_{\infty0}(u,\zeta)\in\mathscr M_{\infty0}.
	\]
	
	It remains to verify the boundary identities. In the
	\((\infty,\infty)\)-sector, we have
	\begin{align*}
		[\zeta^{-1}]\,\mathcal W_{\infty\infty}(\eta,\zeta)
		&=
		uH_{\infty\infty}(u,0)
		=
		N_{\infty\infty}(u,0)
		=
		W_\infty(\eta)^{-1}-\Id_N,
		\\
		[\eta^{-1}]\,\mathcal W_{\infty\infty}(\eta,\zeta)
		&=
		vH_{\infty\infty}(0,v)
		=
		-N_{\infty\infty}(0,v)
		=
		\Id_N-W_\infty(\zeta),
	\end{align*}
	where the normalization condition $W_\infty(\infty)=\Id_N$ has been used.
	Similarly, in the $(0,\infty)$-sector, we have
	\[
	[\zeta^{-1}]\,\mathcal W_{0\infty}(\eta,\zeta)
	=
	N_{0\infty}(\eta,0)
	=
	W_0(\eta)^{-1}-\Id_N.
	\]
	In the $(\infty,0)$-sector, we have
	\[
	[\eta^{-1}]\,\mathcal W_{\infty0}(\eta,\zeta)
	=
	-N_{\infty0}(0,\zeta)
	=
	\Id_N-W_0(\zeta).
	\]
	These calculations establish the boundary identities
	\eqref{eq:boundary-condition} and complete the proof.
\end{proof}

\begin{corollary}[Relative coordinates]
	\label{cor:relative-coordinates}
	The four sectorwise expansions of \(\mathcal W_{ab}\) determine a unique bi-infinite array \(\{w_{ij}\}_{(i,j)\in\bbZ^2}\subset\Mat_N(\bbC)\) such that
	\begin{equation}
		\mathcal W_{ab}(\eta,\zeta)=\varepsilon_a\varepsilon_b
		\sum_{i\in\bbZ_a}\sum_{j\in\bbZ_b}w_{ij}\eta^{-i-1}\zeta^{-j-1}.
		\label{eq:relative-coordinate-expansion}
	\end{equation}
	This array extends the zeroth column \(\{w_{i0}\}_{i\in\bbZ}\) and zeroth row \(\{w_{0j}\}_{j\in\bbZ}\).
\end{corollary}

\begin{proof}
	Each element of \(\mathscr M_{ab}\) has a unique convergent expansion in the monomials $\eta^{-i-1}\zeta^{-j-1}$. Since the four sets \(\bbZ_a\times\bbZ_b\) with \(a,b\in\{0,\infty\}\) form a disjoint partition of \(\mathbb Z^2\), these expansions determine a unique bi-infinite matrix array satisfying \eqref{eq:relative-coordinate-expansion}. The extension property follows immediately from \eqref{eq:dressing-row-expansion}, \eqref{eq:dressing-column-expansion} and \eqref{eq:boundary-condition}.
\end{proof}

We call the entries \(w_{ij}\) the relative coordinates. Indeed, the subfamily \(i,j\geq0\) is the usual array of affine coordinates on the big cell of the Sato Grassmannian \cite{Takasaki1984,Takasaki1985}. Proposition~\ref{prop:kernel-expansion} and Corollary~\ref{cor:relative-coordinates} hence extend the affine coordinates to an array indexed by \((i,j)\in\bbZ^2\).

\begin{theorem}[Recurrence relation]
	\label{thm:relative-coordinate-recurrence}
	The relative coordinates satisfy
	\begin{equation}
		w_{i+1,j}-w_{i,j+1}=w_{i0}w_{0j},\qquad i,j\in\bbZ.
		\label{eq:relative-coordinate-recurrence}
	\end{equation}
\end{theorem}

\begin{proof}
	For \(a,b\in\{0,\infty\}\),
	equation \eqref{eq:relative-kernel-identity} is equivalent to
	\[
	(\zeta-\eta)\mathcal W_{ab}(\eta,\zeta)=W_a(\eta)^{-1}W_b(\zeta)-\Id_N.
	\]
	Substitution of \eqref{eq:dressing-row-expansion}, \eqref{eq:dressing-column-expansion}, and \eqref{eq:relative-coordinate-expansion}, followed by sectorwise coefficient comparison, gives \eqref{eq:relative-coordinate-recurrence}. Indeed, multiplication by \(\eta\) or \(\zeta\) shifts the corresponding index of $w_{ij}$, while \(i=-1\) and \(j=-1\) mark the terms independent of \(\eta\) and \(\zeta\), respectively.
\end{proof}

\section{Sato--theoretic construction of the ASDYM hierarchy}
\label{sec:sato-theoretic-construction}

We now introduce the hierarchy variables and derive the Sato--Wilson equations together with the induced continuous and discrete flows of the relative coordinates.

\subsection{Sato--Wilson equations and coordinate flows}

Let \(\bt=(t_0,t_1,t_2,\ldots)\) be a countable set of hierarchy variables with weights \(\wt(t_n):=n+1\). We allow the dressing matrices to depend on \(\bt\) and set
\begin{equation}
	g(\bt,\lambda):=W_0(\bt,\lambda)^{-1}W_\infty(\bt,\lambda),
	\qquad
	\lambda\in D_0\cap D_\infty.
	\label{eq:time-dependent-transition}
\end{equation}
For \(n\geq1\), introduce \(\mathcal D_n(\lambda):=\partial_n-\lambda\partial_{n-1}\). The Sato--Wilson evolution is governed by
\begin{equation}
	\mathcal D_n(\lambda)g(\bt,\lambda)=0,\qquad n\geq1.
	\label{eq:vacuum-equations}
\end{equation}

\begin{theorem}[Sato--Wilson equations]
	\label{thm:sato-wilson}
	Assume that \(g\) satisfies \eqref{eq:vacuum-equations}. Then, for each \(n\geq1\), there exists a unique \(\lambda\)-independent matrix function \(A_n(\bt)\) such that
	\begin{equation}
		\partial_nW_a(\bt,\lambda)=\lambda\partial_{n-1}W_a(\bt,\lambda)-A_n(\bt)W_a(\bt,\lambda),
		\qquad a\in\{0,\infty\}.
		\label{eq:sato-wilson-equation}
	\end{equation}
\end{theorem}

\begin{proof}
	For later use, set
	\[
	H_{n,a}(\lambda):=\bigl(\mathcal D_n(\lambda)W_a(\lambda)\bigr)W_a(\lambda)^{-1}.
	\]
	Using \eqref{eq:time-dependent-transition} and \eqref{eq:vacuum-equations}, we obtain \(H_{n,0}=H_{n,\infty}\) on \(D_0\cap D_\infty\). The two expressions therefore define a holomorphic matrix function on the Riemann sphere and hence are independent of \(\lambda\). Writing their common value as \(-A_n\) gives \eqref{eq:sato-wilson-equation}; uniqueness is immediate.
\end{proof}

Equations~\eqref{eq:sato-wilson-equation} are the Sato--Wilson
equations of the ASDYM hierarchy \cite{KakeiIkedaTakasaki2002}. The matrices
\(A_n\) determined in
Theorem~\ref{thm:sato-wilson} will be referred to as the
gauge potentials. They admit the following two equivalent
representations \cite{MasonWoodhouse1996}.

\begin{corollary}[Matrix representations]
	\label{cor:matrix-representations}
	For each \(n\geq1\), the gauge potential admits the following representations
	\begin{equation}
		A_n=\partial_{n-1}K=-(\partial_nJ)J^{-1},
		\label{eq:matrix-representations}
	\end{equation}
	where
	\begin{equation}
		J:=\Id_N+w_{0,-1},\qquad K:=-w_{00}.
		\label{eq:J-K-definitions}
	\end{equation}
\end{corollary}

\begin{proof}
	Using the normalized expansion at \(\lambda=\infty\), one has
	\[
	W_\infty=\Id_N-w_{00}\lambda^{-1}+O(\lambda^{-2})
	=\Id_N+K\lambda^{-1}+O(\lambda^{-2}).
	\]
	The coefficient of \(\lambda^0\) in \eqref{eq:sato-wilson-equation} gives \(A_n=\partial_{n-1}K\). The expansion at \(\lambda=0\) implies
	\[
	W_0=J+w_{0,-2}\lambda+O(\lambda^2),
	\]
	and the constant term of \eqref{eq:sato-wilson-equation}  gives \(A_n=-(\partial_nJ)J^{-1}\).
\end{proof}

	Defining the auxiliary linear operator as
	$$\mathcal L_n:=\partial_n-\lambda\partial_{n-1}+A_n,$$
	the
	Sato--Wilson equations take the form $\mathcal L_nW_a=0$. Their
	compatibility condition $[\mathcal L_m,\mathcal L_n]=0$, together
	with \eqref{eq:matrix-representations}, yields the Yang equation
	\cite{Yang1977}
	\begin{equation}
		\partial_{m-1}((\partial_nJ)J^{-1})-\partial_{n-1}((\partial_mJ)J^{-1})=0,
		\label{eq:Yang-equation}
	\end{equation}
	as well as the Leznov--Mukhtarov--Parkes (LMP) equation
	\cite{Leznov1987,LeznovMukhtarov1987,Parkes1992}
	\begin{equation}
		\partial_m\partial_{n-1}K-\partial_n\partial_{m-1}K-[\partial_{n-1}K,\partial_{m-1}K]=0.
		\label{eq:LMP-equation}
	\end{equation}

\begin{remark}
	In this paper, we consider the ASDYM hierarchy with a single sequence of hierarchy variables, in contrast to the formulation associated with toroidal Lie algebras in \cite{KakeiIkedaTakasaki2002}. A suitable specialization of \eqref{eq:Yang-equation} yields Ward's integrable chiral model \cite{Ward1988}, while equation~\eqref{eq:LMP-equation} can be identified with a dispersionless limit of the non-commutative potential Kadomtsev--Petviashvili hierarchy \cite{DimakisMullerHoissen2008}.
	
\end{remark}

\begin{theorem}[Coordinate flows]
	\label{thm:relative-coordinate-flow}
	Under the assumptions of Theorem~\ref{thm:sato-wilson}, the relative coordinates satisfy
	\begin{equation}
		\partial_nw_{ij}-\partial_{n-1}w_{i+1,j}
		=-w_{i0}\partial_{n-1}w_{0j},
		\qquad n\geq1,\quad i,j\in\bbZ.
		\label{eq:relative-coordinate-flow}
	\end{equation}
\end{theorem}

\begin{proof}
	For \(a,b\in\{0,\infty\}\), set
	\[
	M_{ab}(\eta,\zeta):=W_a(\eta)^{-1}W_b(\zeta)
	=N_{ab}(\eta,\zeta)+\Id_N.
	\]
	Using the Sato--Wilson equations \eqref{eq:sato-wilson-equation}, one obtains
		\begin{equation*}
		\partial_n N_{ab}
		=
		-\eta\,
		W_a(\eta)^{-1}
		\bigl(
		\partial_{n-1}W_a(\eta)
		\bigr)
		M_{ab}
		+
		\zeta\,
		W_a(\eta)^{-1}
		\partial_{n-1}W_b(\zeta),
	\end{equation*}
	where the terms involving \(A_n\) cancel. On the other hand,
	\[
	\partial_{n-1}N_{ab}
	=
	-
	W_a(\eta)^{-1}
	\bigl(
	\partial_{n-1}W_a(\eta)
	\bigr)
	M_{ab}
	+
	W_a(\eta)^{-1}
	\partial_{n-1}W_b(\zeta).
	\]
	Hence,
	\[
	\partial_nN_{ab}
	-
	\eta\partial_{n-1}N_{ab}
	=
	(\zeta-\eta)
	W_a(\eta)^{-1}
	\partial_{n-1}W_b(\zeta).
	\]
	It follows from \eqref{eq:relative-kernel-identity} that
	\begin{equation*}
		\partial_n\mathcal W_{ab}(\eta,\zeta)
		-
		\eta\,
		\partial_{n-1}\mathcal W_{ab}(\eta,\zeta)
		=
		W_a(\eta)^{-1}
		\partial_{n-1}W_b(\zeta).
	\end{equation*}
	Substitution of \eqref{eq:dressing-row-expansion}, \eqref{eq:dressing-column-expansion}, and \eqref{eq:relative-coordinate-expansion}, followed by sectorwise coefficient comparison, gives \eqref{eq:relative-coordinate-flow}. Thus the proof is completed.
\end{proof}

The coordinate flows indicate alternative forms of the ASDYM matrix
representations.

\begin{corollary}[Alternative matrix representations]
	\label{cor:alternative-matrix-representations}
	Define
	\begin{equation}
		P:=\Id_N-w_{-1,0},\qquad Q:=w_{-1,-1}.
		\label{eq:P-Q-definitions}
	\end{equation}
	Then \(P=J^{-1}\) and
	\begin{equation}
		\partial_nQ=-(\partial_{n-1}P)P^{-1}=J^{-1}\partial_{n-1}J.
		\label{eq:alternative-matrix-representations}
	\end{equation}
\end{corollary}

\begin{proof}
	Setting \((i,j)=(-1,-1)\) in \eqref{eq:relative-coordinate-recurrence} gives
	\[
	(\Id_N-w_{-1,0})(\Id_N+w_{0,-1})=\Id_N,
	\]
	which is precisely
	$
	PJ=\Id_N,
	$
	and hence \(P=J^{-1}\). The same choice in \eqref{eq:relative-coordinate-flow} yields
	\[
	\partial_nQ=(\Id_N-w_{-1,0})\partial_{n-1}w_{0,-1}
	=J^{-1}\partial_{n-1}J.
	\]
	Differentiating \(P=J^{-1}\) gives the remaining equality.
\end{proof}

Theorem~\ref{thm:relative-coordinate-flow} is compatible with
Corollary~\ref{cor:matrix-representations}. Indeed, setting
$(i,j)=(-1,0)$ in \eqref{eq:relative-coordinate-flow} gives
\[
(I_N-w_{-1,0})\partial_{n-1}w_{00}=\partial_nw_{-1,0},
\]
which, upon using $P=J^{-1}$, implies $\partial_{n-1}K=P^{-1}(\partial_nP)=-(\partial_nJ)J^{-1}$.

\subsection{Dimensional reduction constraints}

Fix a constant matrix \(C\in\Mat_N(\bbC)\), and define
\begin{equation}
	E(\bt,\lambda):=\exp\bigl(-C\xi(\bt,\lambda)\bigr),
	\qquad
	\xi(\bt,\lambda):=\sum_{n\geq0}t_n\lambda^n.
	\label{eq:exponential-vacuum}
\end{equation}
For every \(n\geq0\), differentiation gives
\begin{equation}
	\partial_nE(\bt,\lambda)=-C\lambda^nE(\bt,\lambda),
	\label{eq:vacuum-derivative}
\end{equation}
Hence,
$\mathcal D_n(\lambda)E=0,$
so that \(E\) serves as the vacuum state. Define the Baker functions
\[
\Psi_a(\bt,\lambda):=W_a(\bt,\lambda)E(\bt,\lambda),
\qquad a\in\{0,\infty\}.
\]
Equations \eqref{eq:sato-wilson-equation} and \eqref{eq:vacuum-derivative} give the auxiliary linear system
\begin{equation}
	(\partial_n-\lambda\partial_{n-1}+A_n)\Psi_a=0,
	\qquad n\geq1,
	\label{eq:baker-linear-system}
\end{equation}
and the corresponding transition function is
\begin{equation}
	F(\bt,\lambda):=\Psi_0(\bt,\lambda)^{-1}\Psi_\infty(\bt,\lambda)
	=E(\bt,\lambda)^{-1}g(\bt,\lambda)E(\bt,\lambda).
	\label{eq:baker-transition}
\end{equation}

\begin{theorem}[Dimensional reduction constraints]
	\label{thm:dimensional-reduction}
	Assume that \(F\) is independent of all hierarchy variables. Then
	\begin{equation}
		\partial_0W_a=[W_a,C],\qquad \partial_0\Psi_a=-C\Psi_a,
		\qquad a\in\{0,\infty\}.
		\label{eq:dimensional-reduction-dressing}
	\end{equation}
	In particular,
	\begin{equation}
		\partial_0w_{ij}=[w_{ij},C],\qquad (i,j)\in\bbZ^2.
		\label{eq:dimensional-reduction-coordinates}
	\end{equation}
\end{theorem}

\begin{proof}
	Taking the $t_n$-derivative of $F$ yields
	\[
	\partial_nF=\partial_n(E^{-1}gE)=E^{-1}\left(\partial_ng-\lambda^n[g,C]\right)E.
	\]
	On the other hand, iterating the relations $\mathcal D_n(\lambda)g=0$ gives $\partial_ng=\lambda^n\partial_0g$. Thus $\partial_nF=0$ is equivalent to $\partial_0g=[g,C]$.
	Substituting \(g=W_0^{-1}W_\infty\), we obtain
	\begin{equation}
		(\partial_0W_\infty)W_\infty^{-1}-W_\infty CW_\infty^{-1}
		=(\partial_0W_0)W_0^{-1}-W_0CW_0^{-1}.
		\label{eq:t0-gluing}
	\end{equation}
	The two sides glue to a holomorphic matrix function on the Riemann sphere and are therefore independent of \(\lambda\). The normalization at $\lambda=\infty$ shows that their common value is \(-C\), which gives the first equation in \eqref{eq:dimensional-reduction-dressing}. The second follows from \(\partial_0E=-CE\). Differentiating \eqref{eq:relative-kernel-identity} then gives \(\partial_0\mathcal W_{ab}=[\mathcal W_{ab},C]\), and coefficient comparison then yields \eqref{eq:dimensional-reduction-coordinates}.
\end{proof}

We further emphasize that the four ASDYM matrix variables defined in
\eqref{eq:J-K-definitions} and \eqref{eq:P-Q-definitions} satisfy the same
\(t_0\)-flow constraint:
\begin{equation}
	\partial_0J=[J,C],\qquad \partial_0K=[K,C],\qquad
	\partial_0P=[P,C],\qquad \partial_0Q=[Q,C].
	\label{eq:matrix-t0-reduction}
\end{equation}

\begin{remark}
	For $a\in\{0,\infty\}$,
	the Sato--Wilson equations \eqref{eq:sato-wilson-equation} admit a telescoped form \cite{Takasaki1984,KakeiIkedaTakasaki2002}
	\begin{equation*}
		\partial_nW_a-\lambda^n\partial_0W_a+B_n(\lambda)W_a=0,
		\qquad
		B_n(\lambda):=\sum_{i=1}^n\lambda^{n-i}A_i,
	\end{equation*}
	or equivalently
	\begin{equation*}
		B_n(\lambda)=\bigl(\lambda^n(\partial_0W_\infty)W_\infty^{-1}\bigr)_{\geq0},
	\end{equation*}
	where \((\cdot)_{\geq0}\) denotes projection onto nonnegative powers of \(\lambda\). 
	Substitutingthe $t_0$-flow constraint $\partial_0W_a=[W_a,C]$ gives
	\begin{equation}
		\partial_nW_a-\lambda^nW_aC
		+\bigl(\lambda^nW_\infty CW_\infty^{-1}\bigr)_{\geq0}W_a=0.
		\label{eq:sato-wilson-telescoped}
	\end{equation}
	Likewise, the telescoped form of \eqref{eq:baker-linear-system} is given by
	\begin{equation}
		\partial_n\Psi_a=-(\lambda^nC+B_n(\lambda))\Psi_a,\qquad a\in\{0,\infty\}.
		\label{eq:linear-system-telescoped}
	\end{equation}
\end{remark}

\subsection{Miwa shifts and discrete analogues}

Miwa shifts provide a standard discretization of hierarchy flows \cite{DateJimboMiwa1982,JimboMiwa1983,Miwa1982,DateJimboMiwa1983}. Assume henceforth that $C$ is idempotent
\begin{equation}
	C=\Pi:=\begin{pmatrix}\Id_r&0\\0&0_{N-r}\end{pmatrix},
	\qquad 1\leq r\leq N-1,
	\qquad C^2=C.
	\label{eq:idempotent-C}
\end{equation}
Let \(h\) be a sufficiently small complex parameter such that the shifted Riemann--Hilbert decomposition remains valid. Define
\[
\mathbb T_h:=\exp\left(\sum_{n\geq1}\frac{h^n}{n}\partial_n\right),
\qquad
\mathbb T_hX(\bt)=X(\bt+[h]),
\qquad
[h]:=\left(0,h,\frac{h^2}{2},\frac{h^3}{3},\ldots\right).
\]
For later use, we adopt the abbreviations \(\widetilde X:=\mathbb T_hX\) and \(\Delta_hX:=(\widetilde X-X)/h\).

\begin{lemma}[Miwa shift of the exponential function]
	\label{lem:miwa-exponential}
	For the exponential function defined by \eqref{eq:exponential-vacuum} with the idempotent assumption \eqref{eq:idempotent-C}, we have
	\[
	\widetilde E(\bt,\lambda)=E(\bt,\lambda)(\Id_N-h\lambda C).
	\]
\end{lemma}

\begin{proof}
	The Miwa shift gives
	\[
	E^{-1}\widetilde E
	=\exp\left(-C\sum_{n\geq1}\frac{(h\lambda)^n}{n}\right)
	=\exp\bigl(C\log(1-h\lambda)\bigr).
	\]
	Since \(C^2=C\), one has \(\e^{zC}=\Id_N+(\e^z-1)C\). Taking
	\(z=\log(1-h\lambda)\)  proves the claim.
\end{proof}

\begin{theorem}[Miwa-shift-generated discrete analogue]
	\label{thm:miwa-dressing}
	Assume the hypotheses of Theorems~\ref{thm:sato-wilson} and \ref{thm:dimensional-reduction} hold. Then there exists a unique \(\lambda\)-independent matrix \(U_h\) such that
	\begin{equation}
		\widetilde\Psi_a(\lambda)=L_h(\lambda)\Psi_a(\lambda),
		\qquad
		L_h(\lambda):=\Id_N-h(\lambda C+U_h),
		\qquad a\in\{0,\infty\}.
		\label{eq:discrete-baker-system}
	\end{equation}
	Equivalently,
	\begin{equation}
		\Delta_hW_a(\lambda)
		=\lambda\bigl(\widetilde W_a(\lambda)C-CW_a(\lambda)\bigr)-U_hW_a(\lambda),
		\label{eq:discrete-sato-wilson}
	\end{equation}
	and
	\begin{equation}
		U_h=\widetilde K C-CK=-(\Delta_hJ)J^{-1}.
		\label{eq:discrete-gauge-potential}
	\end{equation}
\end{theorem}

\begin{proof}
	Since \(\widetilde F=F\), the two expressions \(\widetilde\Psi_a\Psi_a^{-1}\) agree on the overlap. By Lemma~\ref{lem:miwa-exponential}, their common value is
	\[
	L_h(\lambda)=\widetilde W_a(\lambda)(\Id_N-h\lambda C)W_a(\lambda)^{-1}.
	\]
	The expression obtained from \(a=0\) is holomorphic on \(D_0\), while the expression obtained from \(a=\infty\) is holomorphic on \(D_\infty\) with at most one simple pole at \(\lambda=\infty\). Hence, \(L_h(\lambda)\) is a matrix polynomial of degree at most one. The expansion at infinity gives
	\[
	L_h(\lambda)=\Id_N+h(CK-\widetilde KC)-h\lambda C,
	\]
	whereas evaluation at \(\lambda=0\) gives \(L_h(0)=\widetilde J J^{-1}\). These identities yield \eqref{eq:discrete-baker-system} and \eqref{eq:discrete-gauge-potential}; equation \eqref{eq:discrete-sato-wilson} follows immediately.
\end{proof}

\begin{remark}
	\label{rem:discrete-asdym-equations}
	Introduce a second Miwa shift $\widehat X:=\mathbb T_kX$ and the associated difference operator
	$\Delta_kX:=(\widehat X-X)/k$. We then consider the pair of linear systems
	\[
	\widetilde\Psi_a(\lambda)=L_h(\lambda)\Psi_a(\lambda),
	\qquad
	\widehat\Psi_a(\lambda)=L_k(\lambda)\Psi_a(\lambda).
	\]
	The compatibility 
	$\widehat{\widetilde\Psi}_a=\widetilde{\widehat\Psi}_a$
	then yields
	$\widehat{L}_h(\lambda)L_k(\lambda)
	=\widetilde{L}_k(\lambda)L_h(\lambda)$.
	Substituting the corresponding matrix expressions for the discrete gauge potentials gives the lattice Yang equation (cf.~\eqref{eq:Yang-equation})
	\begin{equation}
		C(\Delta_hJ)J^{-1}
		-
		C(\Delta_kJ)J^{-1}
		+
		(\Delta_k\widetilde J)\widetilde J^{-1}C
		-
		(\Delta_h\widehat J)\widehat J^{-1}C
		=0,
		\label{eq:full-discrete-j-matrix}
	\end{equation}
	and the lattice LMP equation (cf.~\eqref{eq:LMP-equation})
	\begin{equation}
		\begin{aligned}
			\Delta_h\bigl(\widehat K C-CK\bigr)
			&-
			\Delta_k\bigl(\widetilde K C-CK\bigr) \\
			&+
			\bigl(\widehat{\widetilde K}C-C\widehat K\bigr)
			\bigl(\widehat K C-CK\bigr)
			-
			\bigl(\widehat{\widetilde K}C-C\widetilde K\bigr)
			\bigl(\widetilde K C-CK\bigr)
			=0.
		\end{aligned}
		\label{eq:full-discrete-k-matrix}
	\end{equation}
	Both may be regarded as discrete analogues of the ASDYM equations.
\end{remark}

Theorem~\ref{thm:miwa-dressing} collects the discrete analogues of \eqref{eq:baker-linear-system}, \eqref{eq:sato-wilson-equation}, and \eqref{eq:matrix-representations}.
Indeed, the expansion
\[
\mathbb T_h=\operatorname{id}+h\partial_1+\frac{h^2}{2}(\partial_2+\partial_1^2)+O(h^3)
\]
shows that the first coefficient recovers the \(t_1\)-flow, while the higher coefficients encode compatible combinations of the higher hierarchy flows. 

\begin{theorem}[Discrete coordinate flow]
	\label{thm:discrete-coordinate-flow}
	The relative coordinates satisfy
	\begin{equation}
		\Delta_hw_{ij}-\widetilde w_{i,j+1}C+Cw_{i+1,j}
		=\widetilde w_{i0}Cw_{0j},
		\qquad (i,j)\in\bbZ^2.
		\label{eq:discrete-coordinate-flow}
	\end{equation}
\end{theorem}

\begin{proof}
	For \(\alpha\in\{\eta,\zeta\}\), set \(R(\alpha):=\Id_N-h\alpha C\). Equation \eqref{eq:discrete-sato-wilson} is equivalent to
	\[
	R(\eta)W_a(\eta)^{-1}=\widetilde W_a(\eta)^{-1}L_h(\eta).
	\]
	Using \(L_h(\eta)-L_h(\zeta)=h(\zeta-\eta)C\), we obtain
	\[
	R(\eta)\mathcal W_{ab}(\eta,\zeta)
	-\widetilde{\mathcal W}_{ab}(\eta,\zeta)R(\zeta)
	=h\bigl(\widetilde W_a(\eta)^{-1}CW_b(\zeta)-C\bigr).
	\]
	Rearranging gives
	\begin{equation*}
		\Delta_h\mathcal{W}_{ab}-\zeta\widetilde{\mathcal W}_{ab}C+C\mathcal W_{ab}\eta =
		C-\widetilde W_a(\eta)^{-1}CW_b(\zeta).
	\end{equation*}
	Thus, a substitution of \eqref{eq:dressing-row-expansion}, \eqref{eq:dressing-column-expansion}, and \eqref{eq:relative-coordinate-expansion}, followed by sectorwise coefficient comparison, finally proves \eqref{eq:discrete-coordinate-flow}.
\end{proof}

Likewise, the shift relation \eqref{eq:discrete-coordinate-flow} is a discrete analogue of \eqref{eq:relative-coordinate-flow}.
Setting \((i,j)=(-1,-1)\) in \eqref{eq:discrete-coordinate-flow} therefore gives
\begin{equation}
	\Delta_hQ=C-\widetilde P C P^{-1}=C-\widetilde J^{-1}CJ,
	\label{eq:discrete-Q-representation}
\end{equation}
which serves as the discrete analogue of \eqref{eq:alternative-matrix-representations}.

\subsection{Equivalent $\infty\times\infty$ block matrix representations}
\label{sec:infty-matrix-representation}

The results of Theorems~\ref{thm:relative-coordinate-recurrence},
\ref{thm:relative-coordinate-flow}, and~\ref{thm:discrete-coordinate-flow}
admit compatible representations in terms of $\infty\times\infty$
block matrix algebra.
To this end, let
\begin{equation}
\mathbf w:=
\begin{pmatrix}
	& \vdots & \vdots & \vdots &   \\
	\cdots & w_{-1,-1} & w_{-1,0} & w_{-1,1} & \cdots \\
	\cdots & w_{0,-1} & w_{00} & w_{01} & \cdots \\
	\cdots & w_{1,-1} & w_{10} & w_{11} & \cdots \\
	& \vdots & \vdots & \vdots &   \\
\end{pmatrix}
\label{eq:def-infinity-w}
\end{equation}
be the $\infty\times\infty$ block matrix whose entries are the relative
coordinates $\{w_{ij}\}_{(i,j)\in\mathbb Z^2}$.
Introduce
\[
(\boldsymbol\Lambda)_{ij}:=\delta_{i+1,j}\Id_N,
\qquad
\mathbf{O}:=\delta_{i0}\delta_{0j}\Id_N,
\qquad
\mathbf C:=\operatorname{diag}(\cdots,C,C,C,\cdots),
\]
where
\[
\delta_{ij}=
\begin{cases}
	&1,\qquad i=j, \\
	&0,\qquad i\neq j.
\end{cases}
\]
All products are interpreted as block-matrix products. Every matrix entry appearing below involves only a finite sum.

	Suppose that $(\mathbf A)_{ij}=a_{ij}$ and
	$(\mathbf B)_{ij}=b_{ij}$ are arbitrary $\infty\times\infty$
	matrices. Then
	\[
	(\boldsymbol{\Lambda}\mathbf A)_{ij}=a_{i+1,j},
	\qquad
	(\mathbf A\boldsymbol{\Lambda}^{-1})_{ij}=a_{i,j+1},
	\qquad
	(\mathbf A\mathbf O\mathbf B)_{ij}=a_{i0}b_{0j}.
	\]

\begin{theorem}[Equivalent \texorpdfstring{$\infty\times\infty$}{infinite-matrix} representations]
	\label{thm:infinite-matrix-representations}
	Let $\mathbf w$ be the $\infty\times\infty$ matrix associated with
	the relative coordinates $w_{ij}$. Then it satisfies the recurrence
	relation 
	\begin{equation*}
		\boldsymbol{\Lambda}\mathbf w-\mathbf w\boldsymbol{\Lambda}^{-1}=\mathbf w\mathbf O\mathbf w,
	\end{equation*}
	the continuous coordinate flow
	\begin{equation*}
		\partial_n\mathbf w-\boldsymbol{\Lambda}\partial_{n-1}\mathbf w=-\mathbf w\mathbf O\partial_{n-1}\mathbf w,
	\end{equation*}
	and the discrete coordinate flow
	\begin{equation*}
		\Delta_h\mathbf w-\widetilde{\mathbf w}\boldsymbol{\Lambda}^{-1}\mathbf C+\mathbf C\boldsymbol{\Lambda}\mathbf w=\widetilde{\mathbf w}\mathbf O\mathbf C\mathbf w.
	\end{equation*}

\end{theorem}

\begin{proof}
	The assertions follow entrywise from the preceding statements and equations \eqref{eq:relative-coordinate-recurrence}, \eqref{eq:relative-coordinate-flow}, and \eqref{eq:discrete-coordinate-flow}.
\end{proof}

Takasaki’s early work considered a similar compatible representation, but included only $\{w_{ij}\}_{i,j\geq0}$ \cite{Takasaki1984,Takasaki1985}.
In contrast, the construction \eqref{eq:def-infinity-w} contains $w_{ij}$ with $(i,j)\in\mathbb Z^2$. A result more directly related to Theorem~\ref{thm:infinite-matrix-representations} is the direct linearization scheme for the ASDYM hierarchy developed in \cite{LiZhang2025}.

\section{Integrable reductions and relative-coordinate representations}
\label{sec:Ward-conjecture}

In this section, we show how the Sato-theoretic construction developed above provides a unified framework for such reductions across different formulations of the ASDYM hierarchy.
Motivated by \cite{LiMarunoZhang2026}, we work with the gauge group \(\SL_2(\bbC)\) for convenience, which implies
\begin{equation}
	\det J=1,\qquad \tr K=0,\qquad \det P=1,\qquad \tr Q=0.
	\label{eq:SL2-normalization}
\end{equation}
Condition \eqref{eq:idempotent-C} then gives \(C=\Pi=\diag(1,0)\). Thus for any \(2\times2\) matrix \(X\in\{J,K,P,Q\}\),
\begin{equation}
	\partial_0X=[X,\Pi]
	=\begin{pmatrix}0&-X_{12}\\X_{21}&0\end{pmatrix}.
	\label{eq:Pi-reduction}
\end{equation}
For later use, set $x:=t_1$ and write
\begin{equation}
	K=\begin{pmatrix}K_{11}&K_{12}\\K_{21}&-K_{11}\end{pmatrix}
	=: \begin{pmatrix}u&-q\\r&-u\end{pmatrix},
	\qquad
	J=\begin{pmatrix}J_{11}&J_{12}\\J_{21}&J_{22}\end{pmatrix}
	=:\begin{pmatrix}s&J_{12}\\ps&J_{22}\end{pmatrix}.
	\label{eq:J-K-parametrization}
\end{equation}
In what follows, whenever the operator
\(\partial_x^{-1}:=\int \cdot \, \mathrm d x\) is used, the associated
constant of integration is taken to be zero.
Together with the two Miura-type relations \eqref{eq:matrix-representations} and \eqref{eq:alternative-matrix-representations}, and their discrete counterparts \eqref{eq:discrete-gauge-potential} and \eqref{eq:discrete-Q-representation}, these assumptions provide the basic $\SL_2(\bbC)$ ASDYM reduction framework used below.

We focus on the positive AKNS and GI hierarchies, from which several related integrable models are obtained by gauge transformations. All nonlinear variables appearing in these systems admit relative-coordinate representations.

\subsection{Reductions to continuous AKNS and GI hierarchies}
\label{sec:continuous-AKNS-GI}

The reductions to continuous AKNS and GI are known, we therefore only review them briefly. For a detailed proof, we refer the reader to Theorem~1 and~2 of \cite{LiMarunoZhang2026}.

Recalling the LMP equation \eqref{eq:LMP-equation} and setting $m=1$, we obtain
\begin{equation*}
	\partial_{n-1}\partial_1K-\partial_{n}[K,\Pi]-[\partial_{n-1}K,[K,\Pi]]=0,
	\qquad n\geq1.
\end{equation*}
Expanding this equation shows that the entries satisfy
\begin{equation*}
	u_{x,t_{n-1}}=-(qr)_{t_{n-1}}, \qquad
	q_{x,t_{n-1}}=-q_{t_n}-2u_{t_{n-1}} q, \qquad
	r_{x,t_{n-1}}=r_{t_n}-2u_{t_{n-1}}r.
\end{equation*}
Rearranging these relations yields 
\begin{equation}
	q_{t_{n}}=-q_{x,t_{n-1}}+2q\partial_x^{-1}(rq)_{t_{n-1}},
	\qquad
	r_{t_{n}}=r_{x,t_{n-1}}-2r\partial_x^{-1}(qr)_{t_{n-1}},
	\label{eq:AKNS-hierarchy}
\end{equation}
where $q_{t_0}=-q$ and $r_{t_0}=r$. This is precisely the positive AKNS hierarchy.
The first nontrivial member of \eqref{eq:AKNS-hierarchy} is the coupled nonlinear Schr\"odinger (NLS) system
\begin{equation}
	r_{t_2}=r_{xx}-2qr^2,\qquad q_{t_2}=-q_{xx}+2q^2r.
	\label{eq:NLS-equation}
\end{equation}

We next derive the GI hierarchy through the Miura transformation encoded by the $J$-matrix.
Equation \eqref{eq:matrix-representations} with \(n=1\) expands as
\[
\begin{pmatrix}
	s_x & J_{12,x} \\
	(ps)_x & J_{22,x}
\end{pmatrix}
=
\begin{pmatrix}
	-pqs & -q J_{22} \\
	-rs & -r J_{12}
\end{pmatrix},
\]
whose first column implies
\[
	s_x=-pqs,\qquad
	p_xs+ps_x=-rs.
\]
On a domain where $s\neq0$, eliminating $s$ from the above relations gives 
\begin{equation}
	p_x=-r+p^2q,
	\label{eq:GI-Miura}
\end{equation}
which is the Miura transformation between the AKNS and GI hierarchies \cite{KakeiKikuchi2004}.
Substituting \eqref{eq:GI-Miura} into \eqref{eq:AKNS-hierarchy} yields the positive GI hierarchy
\begin{subequations}
	\label{eq:GI-hierarchy}
	\begin{align}
		p_{t_{n}}&=p_{x,t_{n-1}}-2p\partial_x^{-1}\bigl(pq_x+(pq)^2\bigr)_{t_{n-1}},\qquad p_{t_0}=p,
		\label{eq:GI-p}\\
		q_{t_{n}}&=-q_{x,t_{n-1}}-2q\partial_x^{-1}\bigl(p_xq-(pq)^2\bigr)_{t_{n-1}},\qquad q_{t_0}=-q.
		\label{eq:GI-q}
	\end{align}
\end{subequations}
The first nontrivial flow is given by
\begin{equation}
	p_{t_2}
	=
	p_{xx}
	-2p^2q_x
	-2p^3q^2,
	\qquad
	q_{t_2}
	=
	-q_{xx}
	-2p_xq^2
	+2p^2q^3.
	\label{eq:gi-equation}
\end{equation}

\subsection{Reductions to lattice AKNS and GI systems}

We next derive the corresponding lattice equations from \eqref{eq:discrete-gauge-potential}. Substituting \eqref{eq:J-K-parametrization} gives
\begin{equation}
	\widetilde J=
	\begin{pmatrix}
		1-h(\widetilde u-u)&-hq\\
		-h\widetilde r&1
	\end{pmatrix}J.
	\label{eq:discrete-J-relation}
\end{equation}
Since $\det J=1$ is preserved under Miwa shifts, relation \eqref{eq:discrete-J-relation} then implies
\[
	\det
	\begin{pmatrix}
		1-h(\widetilde{u}-u) & -hq \\
		-h\widetilde r & 1
	\end{pmatrix}=1-h(\widetilde u-u)-h^2q\,\widetilde{r}=1.
\]
Rearranging this relation gives
\begin{equation}
	\Delta_hu=\frac{\widetilde{u}-u}{h}=-q\widetilde{r}.
	\label{eq:discrete-u-h}
\end{equation}
Similarly, recalling a second Miwa shift used in Remark~\ref{rem:discrete-asdym-equations}, one obtains $\Delta_ku=-q\widehat r$.
Substituting \eqref{eq:J-K-parametrization} into \eqref{eq:full-discrete-k-matrix} yields
\begin{subequations}\label{eq:discrete-k-entrywise}
	\begin{align}
		&
		\left(\frac1h-\frac1k\right)
		\bigl(
		\widehat{\widetilde u}
		-\widehat u-\widetilde u+u
		\bigr)
		+
		(\widehat{\widetilde u}-\widehat u)
		(\widehat u-u)
		-
		(\widehat{\widetilde u}-\widetilde u)
		(\widetilde u-u)
		+
		\widehat q\,\widehat r
		-
		\widetilde q\,\widetilde r
		=0,
		\\
		&
		\frac{\widetilde q-q}{h}
		-
		\frac{\widehat q-q}{k}
		+
		(\widetilde u-\widehat u)q
		=0,
		\qquad
		\frac{\widehat{\widetilde r}-\widehat r}{h}
		-
		\frac{\widehat{\widetilde r}-\widetilde r}{k}
		+
		\widehat{\widetilde r}
		(\widehat u-\widetilde u)
		=0.
	\end{align}
\end{subequations}
Using $\Delta_hu=-q\widetilde r$ and $\Delta_ku=-q\widehat{r}$, 
the last two equations in \eqref{eq:discrete-k-entrywise} reduce to the lattice AKNS system
\begin{equation}
		\frac{\widetilde q-q}{h}
		-
		\frac{\widehat q-q}{k}
		=
		q^2(h\widetilde r-k\widehat r), \qquad
		\frac{\widehat{\widetilde r}-\widehat r}{h}
		-
		\frac{\widehat{\widetilde r}-\widetilde r}{k}
		=
		\widehat{\widetilde r}\,^2
		(k\widetilde{q}-h\widehat{q}). 
		\label{eq:discrete-akns}
\end{equation}
The first equation in \eqref{eq:discrete-k-entrywise} is redundant, since it follows algebraically from \eqref{eq:discrete-akns} and the compatibility of \(\Delta_hu=-q\widetilde r\) and \(\Delta_ku=-q\widehat r\).

We next use the discrete Miura relation to obtain the lattice GI system. With \eqref{eq:J-K-parametrization} and \eqref{eq:discrete-u-h}, we rewrite \eqref{eq:discrete-J-relation} as
\begin{equation*}
	\begin{pmatrix}\widetilde s&\widetilde J_{12}\\
		\widetilde p\widetilde s&\widetilde J_{22}\end{pmatrix}
	=
	\begin{pmatrix}1+h^2q\widetilde r&-hq\\-h\widetilde r&1\end{pmatrix}
	\begin{pmatrix}s&J_{12}\\ps&J_{22}\end{pmatrix}.
\end{equation*}
The first column yields
\[
\widetilde s=(1-hpq+h^2q\widetilde r)s,
\qquad
\widetilde p\widetilde s=(p-h\widetilde r)s.
\]
Since \((1-hpq+h^2q\widetilde r)(1+hq\widetilde p)=1\), it follows that
\begin{equation}
	\frac{\widetilde s}{s}=\frac{1}{1+hq\widetilde p} \, ,
	\qquad
	h\widetilde r=p-\frac{\widetilde p}{1+hq\widetilde p}\, .
	\label{eq:discrete-GI-h}
\end{equation}
This is the discrete Miura relation between the lattice AKNS and GI systems. A rearrangement of  \eqref{eq:discrete-GI-h} implies
\[
\Delta_hp
=
-\widetilde r
+pq\widetilde p
-hq\widetilde p\,\widetilde r,
\]
which is a discrete analogue of \eqref{eq:GI-Miura}. Taking the continuous limit $h\to0$ recovers \eqref{eq:GI-Miura}.

Similarly, one has
\begin{equation}
	\frac{\widehat s}{s}=\frac{1}{1+kq\widehat p} \, ,
	\qquad
	k\widehat r=p-\frac{\widehat p}{1+kq\widehat p} \, .
	\label{eq:discrete-GI-k}
\end{equation}
Substituting \eqref{eq:discrete-GI-h} and \eqref{eq:discrete-GI-k} into the first equation of \eqref{eq:discrete-akns} gives
\begin{equation}
	\frac{\widetilde q-q}{h}
	-
	\frac{\widehat q-q}{k}
	=
	q\,^2\left(
	\frac{\widehat p}{1+kq\widehat p}
	-
	\frac{\widetilde p}{1+hq\widetilde p}
	\right).
	\label{eq:full-discrete-gi-q}
\end{equation}

On the other hand, shifting the second relation in
\eqref{eq:discrete-GI-h} in the $k$-direction and the
second relation in \eqref{eq:discrete-GI-k} in the
$h$-direction gives
\[
h\widehat{\widetilde r}
=
\widehat p
-
\frac{\widehat{\widetilde p}}
{1+h\widehat q\,\widehat{\widetilde p}} \, ,
\qquad
k\widehat{\widetilde r}
=
\widetilde p
-
\frac{\widehat{\widetilde p}}
{1+k\widetilde q\,\widehat{\widetilde p}} \, .
\]
These two relations respectively give
\begin{align*}
	\frac{\widehat{\widetilde p}-\widehat p}{h}
	=
	-\widehat{\widetilde r}
	+
	\frac{1}{h}
	\left(
	\widehat{\widetilde p}
	-
	\frac{\widehat{\widetilde p}}
	{1+h\widehat q\,\widehat{\widetilde p}}
	\right)=
	-\widehat{\widetilde r}
	+
	\widehat{\widetilde p}\,^2
	\frac{\widehat q}
	{1+h\widehat q\,\widehat{\widetilde p}} \, ,
\end{align*}
and
\begin{align*}
	\frac{\widehat{\widetilde p}-\widetilde p}{k}
	=
	-\widehat{\widetilde r}
	+
	\frac{1}{k}
	\left(
	\widehat{\widetilde p}
	-
	\frac{\widehat{\widetilde p}}
	{1+k\widetilde q\,\widehat{\widetilde p}}
	\right)=
	-\widehat{\widetilde r}
	+
	\widehat{\widetilde p}\,^2
	\frac{\widetilde q}
	{1+k\widetilde q\,\widehat{\widetilde p}} \, .
\end{align*}
Their difference eliminates $\widehat{\widetilde r}$ and yields 
\begin{equation}
	\frac{\widehat{\widetilde p}-\widehat p}{h}
	-
	\frac{\widehat{\widetilde p}-\widetilde p}{k}
	=
	\widehat{\widetilde p}\,^2
	\left(
	\frac{\widehat q}
	{1+h\widehat q\,\widehat{\widetilde p}}
	-
	\frac{\widetilde q}
	{1+k\widetilde q\,\widehat{\widetilde p}}
	\right).
	\label{eq:full-discrete-gi-p-symmetric}
\end{equation}
Equations
\eqref{eq:full-discrete-gi-q} and
\eqref{eq:full-discrete-gi-p-symmetric}
provide the lattice GI system.

\begin{remark}
	The lattice systems derived above are not merely discrete analogue of the corresponding continuous systems. They are also the functional representations of the corresponding continuous hierarchies  \cite{DimakisMullerHoissen2006,DimakisMullerHoissen2006-2}.
\end{remark}

The integrable systems obtained above admit explicit representations in terms of the relative coordinates. The corresponding representations for the continuous and lattice AKNS/GI systems are summarized in Table~\ref{tab:asdym-reductions-1}.
\begin{table}[htbp]
	\centering
	\small
	\renewcommand{\arraystretch}{2}
	\begin{tabular}{|c|c|c|}
		\hline
		\text{Integrable systems}
		&
		\text{variables in ASDYM matrices}
		&
		\text{associated relative coordinates}
		\\
		\hline
		\text{AKNS}~\eqref{eq:AKNS-hierarchy}/\eqref{eq:discrete-akns}
		&
		$(r,q)=(K_{21},-K_{12})$
		&
		$(-(w_{00})_{21},(w_{00})_{12})$
		\\
		\text{GI}~\eqref{eq:GI-hierarchy}/\eqref{eq:full-discrete-gi-q}\eqref{eq:full-discrete-gi-p-symmetric}
		&
		$\displaystyle(p,q)=\left(\frac{J_{21}}{J_{11}},-K_{12}\right)$
		&
		$\displaystyle\left(\frac{(I_2+w_{0,-1})_{21}}{(I_2+w_{0,-1})_{11}},(w_{00})_{12}\right)$
		\\
		\hline
	\end{tabular}
	\caption{Correspondence between integrable systems and relative coordinates I.}
	\label{tab:asdym-reductions-1}
\end{table}

\subsection{Gauge transformations and related integrable models}

In this subsection, we restrict our discussion to continuous integrable systems; the discrete case can be treated in a similar manner. The AKNS variables arise naturally from the local expansion of the dressing matrix at $\lambda=\infty$. Recalling \eqref{eq:linear-system-telescoped} with $a=\infty$, we obtain
\begin{equation}
	\partial_n\Psi_\infty=-(\lambda^n\Pi+B_n(\lambda))\Psi_\infty,
	\qquad
	B_n:=\sum_{i=1}^n\lambda^{n-i}\partial_{i-1}K.
	\label{eq:linear-system-akns}
\end{equation}
Defining $\Phi^{\text{[AKNS]}}:=e^{\xi/2}\Psi_\infty$ transforms \eqref{eq:linear-system-akns} into the linear system for the AKNS hierarchy
\begin{equation}
	\Phi^{\text{[AKNS]}}_{t_n}=-\left(\lambda^n\frac{\sigma_3}{2}+B_n(\lambda)\right)\Phi^{\text{[AKNS]}}.
\end{equation}
Taking $n=1$ yields the spectral problem for the AKNS hierarchy
\begin{equation}
	\Phi^{\text{[AKNS]}}_x=-
	\left(\lambda \frac{\sigma_3}{2}+[K,\Pi]\right)\Phi^{\text{[AKNS]}}=-\frac{1}{2}
	\begin{pmatrix}
		\lambda & 2q \\
		2r & -\lambda
	\end{pmatrix}\Phi^{\text{[AKNS]}}.
	\label{eq:spectral-akns}
\end{equation}

The NLS flow of the AKNS hierarchy admits a gauge-equivalent representation. Indeed, set (cf.~\cite{WadatiSogo1983})
\[
\Phi^{\text{[HF]}}:=J^{-1}\Phi^{\text{[AKNS]}},
\qquad
S:=J^{-1}\sigma_3 J.
\]
Then $S^2=\Id_2$ and $\tr S=0$.
It follows from \eqref{eq:spectral-akns} that
\[
\begin{aligned}
	\Phi^{\text{[HF]}}_x=(J^{-1}\Phi^{\text{[AKNS]}})_x&=
	-J^{-1}J_x\Phi^{\text{[HF]}}+J^{-1}\Phi^{\text{[AKNS]}}_x
	\\
	&=-J^{-1}J_x\Phi^{\text{[HF]}}-J^{-1}\left(\lambda\frac{\sigma_3}{2}-J_xJ^{-1}\right)J\Phi^{\text{[HF]}}
	=-\frac{\lambda}{2}S\Phi^{\text{[HF]}}.
\end{aligned}
\]
This is precisely the spectral problem for the Heisenberg ferromagnet (HF) model \cite{ZakharovTakhtadzhyan1979}. Likewise, the $t_2$-flow of $\Phi^{\text{[HF]}}$ is given by
\[
\Phi^{\text{[HF]}}_{t_2}=
\left(-\frac{\lambda^2}{2}S+\frac{\lambda}{2}SS_x\right)\Phi^{\text{[HF]}}.
\]
Thus, the compatibility condition $\partial_x\partial_{t_2}\Phi^{\text{[HF]}}=\partial_{t_2}\partial_x\Phi^{\text{[HF]}}$ yields the HF model
\begin{equation}
	S_{t_2}+\frac{1}{2 }
	\left[S,S_{xx}\right]=0.
	\label{eq:hf-model}
\end{equation}
On the other hand, recalling \eqref{eq:alternative-matrix-representations} with $n=1$, we have
\[
\partial_xQ=J^{-1}[J,\Pi]=\frac{1}{2}J^{-1}[J,\sigma_3]=\frac{1}{2}\sigma_3-\frac{1}{2}J^{-1}\sigma_3J=\frac{1}{2}\sigma_3-\frac{1}{2}S.
\]
This gives an alternative representation of $S$ in terms of the $Q$-matrix. The HF model can therefore be expressed in terms of the relative coordinates, as summarized in Table~\ref{tab:asdym-reductions-2}.
\begin{table}[htbp]
	\centering
	\small
	\renewcommand{\arraystretch}{1.5}
	\begin{tabular}{|c|c|c|}
		\hline
		\text{Integrable systems}
		&
		\text{variables in ASDYM matrices}
		&
		\text{associated relative coordinates}
		\\
		\hline
		\text{HF}~\eqref{eq:hf-model}
		&
		$S=J^{-1}\sigma_3 J$
		&
		$(I_2-w_{-1,0})\sigma_3(I_2+w_{0,-1})$
		\\
		\text{HF}~\eqref{eq:hf-model}
		&
		$S=\sigma_3-2Q_x$
		&
		$\sigma_3-2\partial_xw_{-1,-1}$
		\\
		\hline
	\end{tabular}
	\caption{Correspondence between integrable systems and relative coordinates II.}
	\label{tab:asdym-reductions-2}
\end{table}

The GI, Chen--Lee--Liu (CLL), and Kaup--Newell (KN) equations constitute the three standard gauge-equivalent forms of the DNLS equation \cite{WadatiSogo1983,KakeiKikuchi2004}. Indeed, introducing
\begin{equation*}
	\mu(\gamma)
	:=
	s^\gamma p
	=
	J_{21}J_{11}^{\gamma-1},
	\qquad
	\nu(\gamma)
	:=
	s^{-\gamma}q
	=
	-J_{11}^{-\gamma}K_{12},
	\qquad
	\gamma\in\mathbb Z,
\end{equation*}
transforms \eqref{eq:gi-equation} into the generalized derivative NLS (GDNLS) equation
\begin{equation}
	\begin{aligned}
		&\mu_{t_2}=\mu_{xx}+2\gamma \mu\nu\mu_x+2(\gamma-1)\mu^2\nu_x-(\gamma-1)(\gamma-2)\mu^3\nu^2,
		\\
		&\nu_{t_2}=-\nu_{xx}+2\gamma \mu\nu\nu_x+2(\gamma-1)\mu_x\nu^2+(\gamma-1)(\gamma-2)\mu^2\nu^3.
	\end{aligned}
	\label{eq:gdnls-equation}
\end{equation}
Setting $\gamma=0$ recovers the GI equation \eqref{eq:gi-equation}. Setting $\gamma=1$ yields the CLL equation \cite{ChenLeeLiu1979}
\begin{equation}
	\mu_{t_2}=\mu_{xx}+2\mu\nu\mu_x,\qquad
	\nu_{t_2}=-\nu_{xx}+2\mu\nu\nu_x.
	\label{eq:cll-equation}
\end{equation}
Setting $\gamma=2$ gives the KN equation \cite{KaupNewell1978}
\begin{equation}
	\mu_{t_2}=\mu_{xx}+2(\mu^2\nu)_x, \qquad
	\nu_{t_2}=-\nu_{xx}+2(\mu\nu^2)_x.
	\label{eq:kn-equation}
\end{equation}

On the other hand, equation \eqref{eq:alternative-matrix-representations} with $n=1$ implies
\[
\partial_xQ=J^{-1}[J,\Pi]=
\begin{pmatrix}
	-J_{12}J_{21} & -J_{22}J_{12} \\
	J_{11}J_{21} & J_{21}J_{12}
\end{pmatrix},
\]
where $\det J=1$ has been used.
It follows that
\[
\partial_xQ_{21}=J_{11}J_{21}=\mu(2),
\qquad
\partial_x(J_{11}^{-1}J_{12})=J_{11}^{-2}K_{12}=-\nu(2).
\]
In this sense, the pair $(f,g):=(Q_{21},-J_{11}^{-1}J_{12})$ satisfies the potential KN (pKN) equation
\begin{equation}
	f_{t_2}=f_{xx}+2f_x^2g_x,
	\qquad
	g_{t_2}=-g_{xx}+2f_xg_x^2.
	\label{eq:potential-kn}
\end{equation}
This formulation also arises in the construction of the Fokas--Lenells equation \cite{LiLiuZhang2025}. The corresponding relative-coordinate representations are summarized in Table~\ref{tab:asdym-reductions-3}.

\begin{table}[htbp]
	\centering
	\small
	\renewcommand{\arraystretch}{2}
	\begin{tabular}{|c|c|c|}
		\hline
		\text{Integrable systems}
		&
		\text{variables in ASDYM matrices}
		&
		\text{associated relative coordinates}
		\\
		\hline
		\text{GDNLS}~\eqref{eq:gdnls-equation}
		&
		$\displaystyle(\mu,\nu)=\left(J_{21}J_{11}^{\gamma-1},-\frac{K_{12}}{J_{11}^{\gamma} }\right)$
		&
		$\displaystyle\left((I_2+w_{0,-1})_{21}(I_2+w_{0,-1})_{11}^{\gamma-1},
		\frac{(w_{00})_{12}}{(I_2+w_{0,-1})_{11}^\gamma}\right)$
		\\
		\text{CLL}~\eqref{eq:cll-equation}
		&
		$\displaystyle(\mu,\nu)=\left(J_{21},-\frac{K_{12}}{J_{11}}\right)$
		&
		$\displaystyle\left((I_2+w_{0,-1})_{21},\frac{(w_{00})_{12}}{(I_2+w_{0,-1})_{11}}\right)$
		\\
		\text{KN}~\eqref{eq:kn-equation}
		&
		$\displaystyle(\mu,\nu)=\left(J_{21}J_{11},-\frac{K_{12}}{J_{11}^2}\right)$
		&
		$\displaystyle\left((I_2+w_{0,-1})_{21}(I_2+w_{0,-1})_{11}, 	\frac{(w_{00})_{12}}{(I_2+w_{0,-1})_{11}^2}\right)$
		\\
		\text{pKN}~\eqref{eq:potential-kn}
		&
		$\displaystyle (f,g)=\left(Q_{21},-\frac{J_{12}}{J_{11}}\right)$
		&
		$\displaystyle\left((w_{-1,-1})_{21},-\frac{(I_2+w_{0,-1})_{12}}{(I_2+w_{0,-1})_{11}}\right)$
		\\
		\hline
	\end{tabular}
	\caption{Correspondence between integrable systems and relative coordinates III.}
	\label{tab:asdym-reductions-3}
\end{table}

\begin{remark}
	Although the AKNS hierarchy is a classical integrable system, many of the systems discussed above are related to it through suitable gauge transformations \cite{WadatiSogo1983,KakeiKikuchi2004}. However, constructing these transformations and computing the explicit solutions could be rather difficult. This setting avoids this intermediate step by expressing the relevant reductions directly in terms of the ASDYM matrix variables and their associated relative coordinates.
\end{remark}

\subsection{Explicit solutions from the Cauchy matrix structure}

Recently, Takasaki observed that the master functions in Cauchy matrix approach can be identified with affine coordinates on the big cell of the Sato Grassmannian \cite{Takasaki2025}. Here we give an explicit construction of the relative coordinates via the Cauchy matrix approach. For a detailed proof, we refer the reader to several related references \cite{LiQuYiZhang2022,LiQuZhang2023,LiZhang2025}.

We begin by introducing the Cauchy matrix structure.
\begin{definition}[Cauchy matrix structure]
	\label{def:cauchy-matrix-structure}
	The Cauchy matrix structure associated with the ASDYM hierarchy
	is governed by the Sylvester equation
	\begin{equation}
		\Omega(\boldsymbol{t})\Gamma-\Xi\Omega(\boldsymbol{t})=\rho(\boldsymbol{t})\theta(\boldsymbol{t})
		\label{eq:cauchy-sylvester-equation}
	\end{equation}
	together with the compatible dispersion relations
	\begin{equation}
		\partial_n\rho(\boldsymbol{t})=\Xi^n\rho(\boldsymbol{t})\, \Pi,
		\qquad
		\partial_n\theta(\boldsymbol{t})=-\Pi\, \theta(\boldsymbol{t})\Gamma^n.
		\label{eq:dispersion}
	\end{equation}
	We further assume that the solution
	\(\Omega(\boldsymbol{t})\) is invertible on the domain under
	consideration.
	The associated integrable hierarchy is encoded by the master functions
	\begin{equation}
		S^{(i,j)}:=\theta(\boldsymbol{t})\Gamma^i\Omega(\boldsymbol{t})^{-1}\Xi^j\rho(\boldsymbol{t}).
		\label{eq:master-function}
	\end{equation}
	Here, $\Xi,\Gamma\in\GL_M(\mathbb C)$ are spectral matrices,
	while $\rho$ and $\theta^T$ are $M\times 2$ matrix-valued functions.
	Throughout this section, we assume that $\Xi$ and $\Gamma$ have no
	eigenvalues in common, so that the Sylvester equation uniquely
	determines $\Omega(\boldsymbol{t})$ \cite{Sylvester1884}.
\end{definition}

We next prove that the master functions yield the equations of the relative coordinates \eqref{eq:relative-coordinate-recurrence}, \eqref{eq:relative-coordinate-flow}, \eqref{eq:dimensional-reduction-coordinates}.

\begin{proposition}[Recurrence for master functions]
	\label{prop:master-function-recurrence}
	The master functions satisfy the recurrence relation
	\begin{equation}
		S^{(i+1,j)}
		=
		S^{(i,j+1)}+S^{(i,0)}S^{(0,j)},
		\qquad
		i,j\in\mathbb Z.
		\label{eq:cauchy-master-recurrence}
	\end{equation}
\end{proposition}

\begin{proof}
	Multiplying \eqref{eq:cauchy-sylvester-equation}
	from the left and right by $\Omega^{-1}$ yields
	\begin{equation*}
		\Gamma\Omega^{-1}-\Omega^{-1}\Xi
		=
		\Omega^{-1}\rho\theta\Omega^{-1}.
	\end{equation*}
	Multiplying this identity from the left by $\theta\Gamma^i$
	and from the right by $\Xi^j\rho$ gives
	\eqref{eq:cauchy-master-recurrence}.
\end{proof}

\begin{proposition}[Evolution of master functions]
	\label{prop:sato-wilson-cauchy-matrix}
	The master functions satisfy the evolution equation
	\begin{equation}
		\partial_nS^{(i,j)}=\partial_{n-1}S^{(i+1,j)}-S^{(i,0)}\partial_{n-1}S^{(0,j)}.
		\label{eq:evolution}
	\end{equation}
\end{proposition}

\begin{proof}
	Differentiating \eqref{eq:cauchy-sylvester-equation} with respect to $t_n$
	and using \eqref{eq:dispersion} gives
	\[
	\partial_n\Omega=(\partial_{n-1}\Omega)\Gamma-(\partial_{n-1}\rho)\theta.
	\]
	Substituting this relation into \eqref{eq:master-function} and evaluating
	$
	\partial_nS^{(i,j)}-\partial_{n-1}S^{(i+1,j)}
	$
	yields $-S^{(i,0)}\partial_{n-1}S^{(0,j)}$, which proves the assertion.
\end{proof}

\begin{proposition}[$t_0$-flow for master functions]
	\label{prop:cauchy-dimensional}
	The master functions yield the $t_0$-flow dimensional reduction constraint 
	\begin{equation}
		\partial_0S^{(i,j)}=[S^{(i,j)},\Pi].
		\label{eq:master-t0-flow}
	\end{equation}
\end{proposition}

\begin{proof}
	Taking the $t_0$-derivative of
	\eqref{eq:cauchy-sylvester-equation} gives
	$(\partial_0\Omega)\Gamma-\Xi(\partial_0\Omega)=0$, which implies
	$\partial_0\Omega=0$. Direct substitution then yields
	\[
	\partial_0S^{(i,j)}=-\Pi\theta\Gamma^i\Omega^{-1}\Xi^j\rho+\theta\Gamma^i\Omega^{-1}\Xi^j\rho \Pi=[S^{(i,j)},\Pi].
	\]
	This proves \eqref{eq:master-t0-flow}.
\end{proof}

Proposition~\ref{prop:master-function-recurrence}--\ref{prop:cauchy-dimensional} imply that $w_{ij}=S^{(i,j)}$ constitutes a particular solution of the equations governing the relative coordinates.
Define
\begin{equation}
	J:=I_2+S^{(0,-1)},\qquad K:=-S^{(0,0)}, \qquad P:=I_2-S^{(-1,0)}, \qquad Q:=S^{(-1,-1)},
	\label{eq:asdym-variables-cauchy-matrix}
\end{equation}
Then Theorem~\ref{thm:relative-coordinate-flow} and Corollary~\ref{cor:alternative-matrix-representations} imply that
\begin{equation}
	\partial_{n-1}K=-(\partial_nJ)J^{-1},
	\qquad
	\partial_{n}Q=-(\partial_{n-1}P)P^{-1}=J^{-1}\partial_{n-1}J.
	\label{eq:miura-relation-cauchy-matrix}
\end{equation}
We further show that they can be chosen to satisfy the
$\mathrm{SL}_2(\mathbb C)$ gauge conditions \eqref{eq:SL2-normalization}.

\begin{proposition}[Determinant and trace formulae]
	The ASDYM matrix variables satisfy
	\begin{equation}
	\det J = \frac{\det \Gamma}{\det \Xi}, \qquad \operatorname{tr} K = \operatorname{tr} \Xi - \operatorname{tr} \Gamma
		\label{eq:cauchy-JK-central}
	\end{equation}
	and
	\begin{equation}
		\det P
		=
		\frac{\det \Xi}{\det \Gamma},
		\qquad
		\operatorname{tr} Q
		=
		\operatorname{tr} \Xi^{-1}
		-
		\operatorname{tr} \Gamma^{-1}.
		\label{eq:cauchy-PQ-central}
	\end{equation}
\end{proposition}

\begin{proof}
	By \eqref{eq:master-function} and \eqref{eq:asdym-variables-cauchy-matrix}, we have
	$
	J
	=
	I_2+\theta\Omega^{-1}\Xi^{-1}\rho.
	$
	Using the determinant identity
	\[
	\det(I_2+AB)=\det(I_M+BA)
	\]
	for \(A\in\Mat_{2\times M}\) and
	\(B\in\Mat_{M\times 2}\), together with the Sylvester
	equation \eqref{eq:cauchy-sylvester-equation}, we obtain
	\begin{align*}
		\det J
		=
		\det\left(
		I_M+\Omega^{-1}\Xi^{-1}\rho\theta
		\right)
		=
		\det\left(
		\Omega^{-1}\Xi^{-1}\Omega\Gamma
		\right)
		=
		\frac{\det\Gamma}{\det\Xi}.
	\end{align*}
	For \(K=-S^{(0,0)}\), cyclicity of the trace gives
	\begin{align*}
		\operatorname{tr}K
		=
		-\operatorname{tr}
		\left(
		\Omega^{-1}\rho\theta
		\right)
		=
		-\operatorname{tr}\Gamma
		+
		\operatorname{tr}
		\left(
		\Omega^{-1}\Xi\Omega
		\right)
		=
		\operatorname{tr}\Xi
		-
		\operatorname{tr}\Gamma.
	\end{align*}
	These prove \eqref{eq:cauchy-JK-central}. The relations
	\eqref{eq:cauchy-PQ-central} are obtained similarly.
\end{proof}

\begin{remark}
	To satisfy the requirements \eqref{eq:SL2-normalization}, we further set
	\[
	J'=\frac{1}{\sqrt{\operatorname{det}J}}\, J,
	\qquad
	K'=K-\frac{\operatorname{tr}K}{2}I_2, \qquad
	P'=(J')^{-1},
	\qquad
	Q'=Q-\frac{\operatorname{tr}Q}{2}I_2.
	\]
	This normalization does not change the Miura-type relations
	\eqref{eq:miura-relation-cauchy-matrix} and yields
	\eqref{eq:SL2-normalization} exactly.
\end{remark}

The abstract structure in Definition~\ref{def:cauchy-matrix-structure} admits a simple realization in terms of a
dressed Cauchy matrix. Indeed, let $\Xi$ and $\Gamma$ be diagonal matrices,
\[
\Xi=\operatorname{diag}(\kappa_1,\dots,\kappa_M),
\qquad
\Gamma=\operatorname{diag}(\ell_1,\dots,\ell_M),
\]
and let $\rho$ and $\theta$ be of the form
\begin{equation*}
	\rho_{i1}=a_i\, \e^{\xi(\boldsymbol{t},\kappa_i)},\qquad \rho_{i2}=b_i, \qquad
	\theta_{1j}=c_j \, \e^{-\xi(\boldsymbol{t},\ell_j)}, \qquad \theta_{2j}=d_j,
\end{equation*}
where $a_i,b_i,c_j,d_j\in\mathbb C$ for $i,j=1,\ldots, M$.
Then the dispersion relations \eqref{eq:dispersion} are satisfied, while
the Sylvester equation \eqref{eq:cauchy-sylvester-equation} can be solved
entrywise as
\begin{equation*}
	(\Omega)_{ij}=-\frac{a_ic_j \, \e\,^{\xi(\boldsymbol{t},\kappa_i)-\xi(\boldsymbol{t},\ell_j)}+b_id_j}{\kappa_i-\ell_j}.
\end{equation*}
Thus, $\Omega$ is simply the Cauchy kernel
$(\kappa_i-\ell_j)^{-1}$ dressed by the functions entering
$\rho$ and $\theta$. These functions therefore determine the explicit
form of $S^{(i,j)}$ through \eqref{eq:master-function}.

\section{Concluding remarks}
\label{sec:concluding-remarks}

In this paper, we have developed a unified Sato--theoretic formulation of the ASDYM hierarchy based on a normalized Riemann--Hilbert decomposition. The local dressing matrices at $\lambda=0$ and $\lambda=\infty$ are encoded by a bi-infinite array of relative coordinates, within which the ASDYM matrix formulations and their associated integrable reductions can be treated uniformly.

A key ingredient is the generating function \eqref{eq:relative-kernel-identity}, whose four-sector expansions determine the bi-infinite array $\{w_{ij}\}_{(i,j)\in\mathbb Z^2}$. In particular, expanding \eqref{eq:relative-kernel-identity} at $(\eta,\zeta)=(\infty,\infty)$ yields
\[
\frac{W_\infty(\eta)^{-1}W_\infty(\zeta)-I_N}{\zeta-\eta}
=
\sum_{i\geq0}\sum_{j\geq0}
w_{ij}\eta^{-i-1}\zeta^{-j-1},
\]
which recovers the coefficients $w_{ij}$ for $i,j\geq0$. Thus, the full family $\{w_{ij}\}_{(i,j)\in\mathbb Z^2}$ may be viewed as a natural extension of the affine coordinates on the big cell of the Sato Grassmannian. The inclusion of negative indices substantially enlarges this coordinate framework: while $w_{00}$ naturally leads to AKNS-type equations, the full bi-infinite family provides a uniform description of a broader class of integrable systems.

Within this framework, the integrable systems considered here admit direct realizations in terms of the relative coordinates; see Tables~\ref{tab:asdym-reductions-1}--\ref{tab:asdym-reductions-3} in Section~\ref{sec:Ward-conjecture}. We have also derived several lattice equations through Miwa shifts, providing an alternative route to functional representations of the same hierarchy. Finally, the Cauchy matrix construction gives explicit Grammian-type realizations of the relative coordinates.

For computational convenience, we have imposed several simplifying assumptions and considered a number of representative cases. These choices suggest several directions for further development.

\begin{enumerate}
	\item Throughout Section~\ref{sec:Ward-conjecture}, we restrict our attention to the $\SL_2(\mathbb C)$ gauge group. A natural extension is to replace it by $\GL_N(\mathbb C)$, which would allow reductions to multi-component integrable systems to be considered within the same framework.

	\item In the present paper, we have considered the ASDYM hierarchy with a single sequence of time variables \cite{DimakisMullerHoissen2008}. In fact, a more general formulation should incorporate multiple sequences of hierarchy times, thereby allowing reductions associated with integrable systems possessing toroidal Lie algebra structures \cite{Takasaki1990,KakeiIkedaTakasaki2002}.
	
	\item In Section~\ref{sec:Ward-conjecture}, only the lattice AKNS and GI systems have been studied explicitly. However, Tables~\ref{tab:asdym-reductions-1}--\ref{tab:asdym-reductions-3} already suggest corresponding realizations of the lattice HF and DNLS systems. It would therefore be natural to derive these discrete systems systematically and clarify their interrelations within the same ASDYM framework.
	
	\item For explicit solutions, we have focused on the Cauchy matrix realization \cite{LiQuYiZhang2022,LiQuZhang2023,LiZhang2025}, which yields Grammian-type solutions. Other constructions, such as Wronskian representations \cite{HamanakaHuang2022,Huang2021}, should likewise admit formulations in terms of the relative coordinates, further connecting the present Sato-theoretic framework with standard solution techniques.
	
\end{enumerate}

Taken together, these results provide a Sato--theoretic basis for further study of the ASDYM hierarchy and its continuous and discrete reductions, and offer a concrete point of departure for investigating the Ward conjecture.

\section*{Acknowledgments}
The authors are grateful for valuable discussions with Masashi Hamanaka, Saburo Kakei and
Kanehisa Takasaki during the Mini-Workshop on ASD Yang--Mills Equations at
Nagoya University.  The work of SSL is supported by the JSPS Overseas
Research Fellowships (P25326). The work of DJZ is supported by the NSFC grant (12271334, 1241540016)

\end{document}